\documentclass[12pt]{article}

\usepackage[utf8]{inputenc}
\usepackage[T1]{fontenc}

\usepackage{amsmath,amssymb, amsthm}
\usepackage{hyperref}
\usepackage{stmaryrd}
\usepackage{enumerate}
\usepackage{enumitem}
\usepackage[algoruled,vlined,english,linesnumbered]{algorithm2e}
\usepackage{comment}
\usepackage{multirow}
\usepackage{tikz}
\newcommand{\noopsort}[1]{}

\definecolor{purple}{rgb}{0.6,0,0.6}

\def\binom#1#2{\Big(\begin{array}{cc} #1 \\ #2 \end{array}\Big)}

\title{A Tropical F5 algorithm}
\author{       
        Tristan Vaccon\footnote{Université de Limoges, tristan.vaccon@unilim.fr}, Kazuhiro Yokoyama\footnote{Rikkyo University, kazuhiro@rikkyo.ac.jp} 
}
\date{2017}

\begin{document}

\newtheorem{theo}{Theorem}[section]
\newtheorem{lem}[theo]{Lemma}
\newtheorem{prop}[theo]{Proposition}
\newtheorem{cor}[theo]{Corollary}
\newtheorem{quest}[theo]{Question}
\newtheorem{conj}[theo]{Conjecture}
\theoremstyle{definition}
\newtheorem{rmk}[theo]{Remark}
\newtheorem{ex}[theo]{Example}
\newtheorem{deftn}[theo]{Definition}

\maketitle

\begin{abstract}
Let $K$ be a field equipped with a valuation. Tropical varieties over $K$ can be defined with a theory of Gröbner bases taking into account the valuation of $K$.
While generalizing the classical theory of Gröbner bases, it is not clear how modern algorithms for computing Gröbner bases 
can be adapted to the tropical case.
Among them, one of the most efficient is the celebrated F5 Algorithm of Faugère.

In this article, we prove that, for homogeneous ideals, it can be adapted to the tropical case. We prove termination and correctness.
 
Because of the use of the valuation, the theory of tropical Gröbner bases is promising for stable computations over polynomial rings over a $p$-adic field. We provide numerical examples to illustrate
time-complexity and $p$-adic stability of this tropical F5 algorithm.
\end{abstract}

\section{Introduction}

The theory of tropical geometry is only a few decades old. It has nevertheless already proved to be of significant value, with applications in algebraic geometry, combinatorics, computer science, and non-archimedean geometry (see \cite{MS:2015}, \cite{EKL06}) and even attempts at proving the Riemann hypothesis (see \cite{Connes:2015b}).

Effective computation over tropical varieties make decisive usage of Gröbner bases, but before Chan and Maclagan's definition of tropical Gröbner bases taking into account the valuation in \cite{Chan:2013,CM:2013}, computations were only available over fields with trivial valuation where standard Gröbner bases techniques applied.

In this document, we show that following this definition, an F5 algorithms can be performed to compute tropical Gröbner bases.

Our motivations are twofold. 
Our first objective is to provide with the F5 algorithm a fast algorithm for tropical geometry purposes.
Indeed, for classical Gröbner bases, the F5 algorithm, along with the F4, is recognized as among the fastest available nowadays.

Secondly, while our algorithms
are valid and implemented for computations over $\mathbb{Q},$ 
we aim at computation over fields with valuation that might not be effective, such as $\mathbb{Q}_p$ or $\mathbb{Q} (( t )).$ 
Indeed, in \cite{Vaccon:2015}, the first author has studied the computation of tropical Gröbner bases over such fields
through a Matrix-F5 algorithm. 
For some special term orders, the numerical stability is then remarkable. 
Hence, our second objective in designing a tropical F5 algorithm is to pave the way for an
algorithm that could
at the same time be comparable to the fast methods for classical Gröbner bases, have a termination criterion and 
still benefit from the stability that can be obtained for tropical Gröbner bases.

\vspace{-.2cm}

\subsection{Related works on tropical Gröbner bases} We refer to the book of Maclagan and Sturmfels \cite{MS:2015} for an introduction to computational tropical algebraic geometry.

The computation of tropical varieties over $\mathbb{Q}$ with trivial valuation is available in the Gfan package by Anders Jensen (see \cite{Jensen}), by using standard Gröbner bases computations.
Yet, for computation of tropical varieties over general fields, with non-trivial valuation, such techniques are not readily available. 
Then Chan and Maclagan have developed in \cite{CM:2013} a way to extend the theory of Gröbner bases to take into account the valuation and allow tropical computations. Their theory of tropical Gröbner bases is effective and allows, with a suitable division algorithm, a Buchberger algorithm.
Following their work, a Matrix-F5 algorithm has been proposed in \cite{Vaccon:2015}.

\vspace{-.2cm}

\subsection{Main idea and results}

Let $G'$ be a finite subset of homogeneous polynomials of $A:=k[X_1,\dots, X_n]$
for $k$ a field with valuation.\footnote{\textit{e.g.} $\mathbb{Q}$ or $\mathbb{Q}_p$
with $p$-adic valuation, or $\mathbb{Q}[[X]]$ with $X$-adic} We assume that $G'$
is a tropical Gröbner basis of the ideal $I'$ it spans for a given
tropical term order $\leq.$
Let $f_1 \in A$ be homogeneous. We are interested in computing a tropical 
Gröbner basis of
$I=I'+\left\langle f_1 \right\rangle.$
In this homogeneous context, following Lazard \cite{Lazard:1983}, we can
perform computations in $I_d=I \cap k[X_1,\dots, X_n]_d.$
$I_d$ can be written as a vector space as
$I_d=\left\langle x^\alpha f_1, \: \vert x^\alpha \vert + \vert f_1 \vert =d \right\rangle +I_d' ,$
with $\vert \: \vert$ denoting total degree.
The second summand is already well-known as $G'$ is a tropical Gröbner basis.
Thanks to this way of writing the first summand, we can then
filtrate the vector space $I_d$ by ordering the possible 
$x^\alpha.$
The main idea of the F5 algorithm of Faugère \cite{F5}
is to use knowledge of this filtration
to prevent unnecessary computations.
Our main result is then:

\begin{theo} \label{thm_intro}
The Tropical F5 algorithm (Algorithm \ref{F5_algo} on page \pageref{F5_algo})
computes a tropical Gröbner basis of $I.$
If $f_1$ is not a zero-divisor in $A/I'$ then no
polynomial reduces to zero during the computation.
\end{theo}

\subsection{Notations}
Let $k$ be a field with valuation  $val.$
The polynomial ring $k[X_1,\dots, X_n]$ will be denoted by $A.$ Let $T$ be the set of monomials of $A.$
For $u=(u_1,\dots,u_n) \in \mathbb{Z}_{\geq 0}^n$, we write $x^u$ for
$X_1^{u_1} \dots X_n^{u_n}$ and $\vert x^u \vert$ for its degree.
For $d \in \mathbb{N},$ $A_d$ is the vector space of polynomials in $A$
of degree $d.$ Given a mapping $f : U \rightarrow V,$ $Im(f)$
denote the image of $f.$ For a matrix $M,$ $Rows(M)$ is the list
of its rows, and $Im(M)$ denotes
the left-image of $M$ (\textit{n.b.} $Im(M)=span(Rows(M)$).
For $w \in Im(val)^n \subset \mathbb{R}^n$ and $\leq_1$ a monomial order on $A,$
we define $\leq$ a tropical term order as in the following definition:

\begin{deftn} \label{defn:trop_term_order}
Given $a,b \in k^*$ and $x^\alpha$ and $x^\beta$ two monomials in $A$, we write $a x^\alpha \geq b x^\beta$ if $val(a)+w \cdot \alpha < val(b) +w \cdot \beta$, or $val(a)+w \cdot \alpha = val(b) +w \cdot \beta$ and $x^\alpha \geq_1 x^\beta$.
\end{deftn}

This is a total term order on $A.$ We can then define accordingly for 
$f \in A$ its highest term, denoted by $LT(f),$ and the corresponding
monomial $LM(f).$ These defintions extend naturally to 
$LM(I)$ and $LT(I)$ for $I$
an ideal of $A.$ A tropical Gröbner bases of $I$ (see \cite{CM:2013, Vaccon:2015}) is then
a subset of $I$ such that its leading monomials generate as a monoid $LM(I).$
We denote by $NS(I)$ the set of monomials $T \setminus LM(I).$
We will write occasionally \textit{tropical GB}.

Let $G'$ be a finite subset of homogeneous polynomials of $A$
that is a tropical Gröbner basis of the ideal $I'$ it spans.
Let $f_1 \in A$ be homogeneous. We are interested in computing a tropical Gröbner basis of
$I=I'+\left\langle f_1 \right\rangle.$

\subsection*{Acknowledgements}

We thank Jean-Charles Faugère, Pierre-Jean Spaenlehauer, Masayuki Noro, Naoyuki Shinohara and Thibaut Verron
for fruitful discussions.

\section{Signature}

Contrary to the Buchberger or the F4 algorithm, the F5 algorithm
relies on tags attached to polynomial so as
to avoid unnecessary computation thanks to
this extra information.
Those tags are called \textit{signature}, and they
are decisive for the F5 criterion, which is one of 
the main ingredients of the F5 algorithm.

In this section, we provide a definition for the notion
of signature we need for the F5 algorithm
and deduce some of its first properties.

\begin{deftn}[Syzygies]
Let $v$ be the following $k$-linear map defined by the multiplication by $f_1$:
\begin{equation*}
      v : \: \:      
     \begin{array}{ccl}
      A & \rightarrow & A/I'\\
      1 & \mapsto & f_1.\\
     \end{array}
\end{equation*}
Let $Syz_{f_1}=Ker (v),$ $LM(Syz_{f_1})$ be the leading monomials of the 
polynomials in  $Syz_{f_1}$ and $NS( Syz_{f_1})=T \setminus LM(Syz_{f_1})$
the normal set of monomials for the module of syzygies.
\end{deftn}
\begin{prop} \label{prop:rem_syz}
$\left\langle v(NS(Syz_{f_1})) \right\rangle=Im(v).$
\end{prop}
\begin{proof}
We can prove this claim degree by degree. The method is quite
similar to what was developped in Section 3.2 of \cite{Vaccon:2015}.
Let $d \in \mathbb{N}.$
Let $\alpha_1,\dots,\alpha_u$ be $LM(Syz_{f_1})\cap A_d.$
Let $f_{\alpha_1},\dots,f_{\alpha_u} \in A_d$
be such that for all $i,$  $f_{\alpha_i}\in Ker(v)$
and $LM(f_{\alpha_i})=\alpha_i.$
Then, written in the basis $U_d=\left( LM(Syz_{f_1})\cap A_d \right) \cup \left( NS(Syz_{f_1})\cap A_d \right),$
$f_{\alpha_1},\dots,f_{\alpha_u}$ is under
(row)-echelon form.
Let $g_{\alpha_1},\dots,g_{\alpha_u} \in A_d$ be the corresponding
reduced row-echelon form: we have $LM(g_{\alpha_i})=\alpha_i$, 
$g_{\alpha_i} \in Ker(v)$ and the only monomial 
of $g_{\alpha_i}$ not in $NS(Syz_{f_1})$
is $\alpha_i.$
Now, clearly, if $f\in A_d,$ by reduction by the $g_{\alpha_i}$'s,
there exists
some $g \in A_d$ such that $v(f)=v(g)$
and $LM(g) \in NS(Syz_{f_1}).$  
\end{proof}
\begin{rmk}
Clearly, $LM(I') \subset LM(Syz_{f_1}),$ and this is an equality if $f_1$ 
is not a zero-divisor in $A/I'.$
\end{rmk}

We can now proceed to define the notion of signature.
It relies on a special order on the monomials
of $A$, which is not a monomial order:

\begin{deftn}
Let $x^\alpha$ and $x^\beta$ be monomials in $T.$
We write that $x^\alpha \leq_{sign} x^\beta$ if:
\begin{enumerate}
\item if $\vert x^\alpha \vert < \vert x^\beta \vert.$
\item if $\vert x^\alpha \vert = \vert x^\beta \vert,$ 
$x^\alpha \in NS(Syz_{f_1})$
and $x^\beta \notin NS(Syz_{f_1}).$
\item if $\vert x^\alpha \vert = \vert x^\beta \vert,$
 $x^\alpha, x^\beta \in NS(Syz_{f_1})$ and $x^\alpha \leq x^\beta$
\item if $\vert x^\alpha \vert = \vert x^\beta \vert,$
$x^\alpha, x^\beta \notin NS(Syz_{f_1})$ and $x^\alpha \leq x^\beta.$ 
\end{enumerate}
\end{deftn}
\begin{prop}
$\leq_{sign}$ defines a total order on $T.$
It is degree-refining.
At a given degree, any $x^\alpha$ in $LM(I')$ or $LM(Syz_{f_1})$ is 
bigger than any $x^\beta \notin LM(Syz_{f_1}).$
\end{prop}

We can define naturally $LM_{sign}(g)$ for any $g \in A.$
We should remark that in the general case, $LM_{sign}(g) \neq LM(g).$

\begin{deftn}[Signature]
For $p\in I,$ using the convention that $LM_{sign}(0)=0,$
 we define the signature of $p,$ denoted by $S(p),$ to be
\[S(p)= \min_{\leq_{sign}} \{ LM_{sign}(g_1) \text{ for } g_1 \in A \text{ s.t. } (g_1 f_1-p) \in I' \} .\]
\end{deftn}
\begin{prop}
The signature is well-defined.
\end{prop}
\begin{rmk}
Clearly, $p \in I'$ if and only if $S(p)=0.$
Similarly, if $f_1 \notin I'$ then $S(f_1)=1.$
\end{rmk}
\begin{rmk}
This definition is an extension to the tropical case 
of that of \cite{JCFnoJNCF}.
For trivial valuation, it coincides (after projection on last component) with  that of \cite{Arri-Perry, F5} for elements in $I \setminus I'.$ We have modified
it to ensure that the signature takes value in 
$ \left( T \setminus LM(I') \right) \cup \left\lbrace 0 \right\rbrace,$
 see Prop. \ref{prop:im_of_sign} below.  
\end{rmk}
With the fact that we can decompose an equality by degree, we have the following lemma:
\begin{lem}
If $p \in I \setminus I'$ is homogeneous of degree $d,$ then 
$\deg ( S(p))=\deg (p)-\deg (f_1).$ \label{lem:deg_of_sign}
\end{lem}

\begin{prop} \label{prop:im_of_sign}
For any $p \in I \setminus I',$ $S(p) \in NS( Syz_{f_1}).$
\end{prop}
\begin{proof}
This is a direct consequence of 
Proposition \ref{prop:rem_syz}.
\end{proof}
This proposition can be a little refined.

\begin{lem} \label{lem:sign_comp_mult}
Let $t \in T,$ $f \in I \setminus I'$ then
\begin{eqnarray*}
S(tf)=tS(f) & \Leftrightarrow & tS(f) \in NS(Syz_{f_1}) \\
S(tf) <_{sign}tS(f)& \Leftrightarrow & tS(f) \in LM(Syz_{f_1}). \\ 
\end{eqnarray*}
\end{lem}
\begin{proof}
Let $S(f)=\sigma.$
By definition of $S,$ we have $f=\alpha \sigma f_1+gf_1+h$
with $\alpha \in k^*=k \setminus \{0 \},$ $LM(g)<_{sign} \alpha \sigma$ and $h \in I'.$
If $t \sigma \in NS(Syz_{f_1})$, then we get directly that the leading term of the normal form 
(modulo $G'$) of 
$\alpha t\sigma f_1+tgf_1+th$ is $t\sigma$ and in this case $S(tf)=t \sigma=tS(f),$ otherwise we can provide
a syzygy.
In the other case, the normal form has a strictly smaller leading term, and we get that $S(tf) <_{sign}tS(f).$
\end{proof}

We then have the following property, and two last lemmas to understand the behaviour of signature.

\begin{prop} \label{prop:phi_bijection}
The following mapping $\Phi$ is a \textbf{bijection}:      
\begin{equation*}
      \Phi : \: \:      
     \begin{array}{ccl}
      LM(I) \setminus LM(I') & \rightarrow & NS(Syz_{f_1})\\
      x^\alpha & \mapsto & \min_{\leq_{sign}} \{ S(x^\alpha+g), \: \text{ with } g \text{ s.t. }  \\
      &&x^\alpha+g \in I \text{ and } LT(g)<x^\alpha\},\\
     \end{array}
\end{equation*}
\end{prop}
\begin{proof}
It can be proved directly by computing a tropical row-echelon form of some
Macaulay matrix (see Definition \ref{def:trop_row_echelon_form}).
\end{proof}


\begin{lem} \label{lem:baisse_signature}
Let $f_1,f_2 \in I \setminus I'$ be such that $LM(f_1) > LM(f_2)$ and $S(f_1)=S(f_2)=\sigma.$
Then there exist $\alpha, \beta \in k^*$ such that
\[S(\alpha f_1 + \beta f_2) <_{sign} \sigma. \]
\end{lem}
\begin{lem} \label{lem:addition_signature}
Let $f_1,f_2 \in I$ be such that $S(f_1)>S(f_2).$
Then for any $\alpha \in k^*,\beta \in k,$ $S(\alpha f_1 + \beta f_2) = S(f_1).$
\end{lem}

\section{Tropical $\mathfrak{S}$-Gröbner bases}


The notion of signature allows the definition of a natural filtration of the
vector space $I$ by degree and signature:
\begin{deftn}[Filtration by signature]
For $d \in \mathbb{Z}_{\geq 0}$  and $x^\alpha \in T \cap I_d,$ 
we define the vector space 
\[I_{d, \leq x^\alpha}:=\left\lbrace f \in I_d, \: S(f) \leq x^\alpha \right\rbrace.\]
Then we define the filtration by signature, 
$I=\left( I_{d, \leq x^\alpha} \right)_{d \in \mathbb{Z}_{\geq 0}, \: x^\alpha \in T \cap I_d}.$ \label{def:filtration}
\end{deftn}

Our goal in this section and the following is to define tropical Gröbner bases
that are compatible with this filtration by signature. It relies
on the notion of $\mathfrak{S}$-reduction and irreducibility.

\begin{deftn}[$\mathfrak{S}$-reduction]
Let $f,g \in I,$ $h \in I$ and $\sigma \in T.$
We say that $f$ 
$\mathfrak{S}$\textbf{-reduces} 
to $g$ with respect to $\sigma$ and with $h,$
\[f \rightarrow^h_{\mathfrak{S}, \sigma} g \]
if there are $t \in T$ and $\alpha \in k^*$ such that:
\begin{itemize}
\item $LT(g)<LT(f),$ $LM(g) \neq LM(f)$ and $f-\alpha th=g$ and
\item $S(th)<_{sign} \sigma.$
\end{itemize}
If $\sigma$ is not specified, then we mean $\sigma = S(f).$ 
\end{deftn}


It is then natural to define what is an $\mathfrak{S}$-irreducible polynomial.

\begin{deftn}[$\mathfrak{S}$-irreducible polynomial]
We say that $f \in I$ is $\mathfrak{S}$-irreducible with respect 
to $\sigma \in T,$ or up to $\sigma,$ if there is no $h \in I$ which $\mathfrak{S}$-reduces it
with respect to $\sigma.$
If $\sigma$ is not specified, then we mean $\sigma = S(f).$ If there is no
ambiguity, we might omit the $\mathfrak{S}-.$
\end{deftn}
\begin{rmk}
This definition clearly depends on $I,$ $I',$ and the given monomial ordering.
\end{rmk}

In order to better understand what are $\mathfrak{S}$-irreducible polynomials, we
have the following:

\begin{theo}
Let $f \in I$ and $x^\alpha \in T$ such that:
\begin{itemize}
\item $f$ is $\mathfrak{S}$-irreducible with respect to $x^\alpha,$
\item $f=(c x^\alpha +u)f_1+h$ with $c \in k,$ $u \in A$ with $LT(u) <cx^\alpha,$ $h \in I'.$
\end{itemize}
Then $f=0$ if and only if $x^\alpha \in LM(Syz_{f_1}).$
Moreover, if $f \neq 0,$ then $x^\alpha=S(f)=\Phi (LM(f)).$

\label{thm:alternative_irred}
\end{theo}
\begin{proof}
Suppose $f=0,$ then $0=f=(c x^\alpha +u)f_1+h,$ hence $x^\alpha \in LM(Syz_{f_1}).$
We prove the converse result by contradiction. 
We assume that $f \neq 0$ and $x^\alpha \in LM(Syz_{f_1}).$ 
Let $t = LM(f)$ and $x^\beta = \Phi (t) \in NS(Syz_{f_1}).$
We have $x^\beta <_{sign} x^\alpha.$ 
There exists $g \in I$ such that $LM(g)=t$ and $S(g)=x^\beta.$ 
Since $x^\beta < x^\alpha,$ $g$ is an $\mathfrak{S}$-reductor for $f$ with respect to $x^\alpha.$ 
This contradicts the fact that $f$ is irreducible. Hence $f=0.$
For the additional fact when $f \neq 0,$ ssume $x^\alpha \in NS(Syz_{f_1}).$ 
Then necessarily, $x^\alpha = S(f).$ Suppose now that $x^\alpha \neq \Phi (LM(f)).$ 
Then $S(f)>\Phi(LM(f) ).$ Therefore there exists a polynomial $g \in I$ such that
$t=LM(g)=LM(f)$ and $\Phi (t) =S(g)=\Phi(LM(f))<S(f).$
It follows that $g$ is an $\mathfrak{S}$-reductor of $f,$ which leads to
$f$ being not $\mathfrak{S}$-irreducible.
\end{proof}

\begin{cor}
If $g \in I$ and $x^\alpha \in T$ are such that $x^\alpha g \neq 0$ is $\mathfrak{S}$-irreducible
up to $x^\alpha S(g)$
then $S(x^\alpha g)=x^\alpha S(g).$ \label{cor:sign_prod_irred}
\end{cor}

The previous results show that the right notion of 
$\mathfrak{S}$-irreducibility for $f$ a polynomial 
is up to $S(f).$
Nevertheless, the previous corollary can not be used for 
easy computation of the signatures of irreducible polynomials
as we can see on the following example:

\begin{ex} \label{exp:prob_sign}
We assume that $z^4 \in G'$ and $z^3,x^4$ and $x^2 y^2 \in NS(I').$
We assume that we have $h_1=xy^5+y^5z,$ $h_2 = x^3 y^2 z-y^6$
and $h_3=x^4 y^2+y^6$ such that $S(h_1)=x^2y,$
$S(h_2)= x^3$ and $S(h_3)=z^3$ and all of them are
$\mathfrak{S}$-irreducible. We assume that
$x^4 >_{sign} x^2 y^2$ and $\Phi(x^4y^2z)=x^4.$
Then, $zh_3=yh_1+xh_2=x^4y^2z+y^6z.$
Its signature is $xS(h_2)=x^4,$ and not $z^4.$
With our assumptions, the polynomial $zh_3$ is irreducible
up to $S(zh_3)=x^4,$ whereas up to $z^4,$ it is not.

In other words, it is possible that the polynomial
$x^\alpha f$ is irreducible, up to $S(x^\alpha f),$
 even though $S(x^\alpha f) < x^\alpha S(f).$
\end{ex}

We now have enough definitions to write down the notion of $\mathfrak{S}$-Gröbner bases,
which will be a computational key point for the F5 algorithm.

\begin{deftn}[Tropical $\mathfrak{S}$-Gröbner basis]
We say that $G \subset I$ is a \textbf{tropical} $\mathfrak{S}$\textbf{-Gröbner basis} (or tropical $\mathfrak{S}-$GB, or just $\mathfrak{S}-$GB for short when there is no amibuity) of $I$ with respect to $G',$ $I',$ and a given tropical term order if 
$G'=\{g \in G \text{ s.t. } S(g)=0 \}$ and
for each $\mathfrak{S}$-irreducible polynomial $f \in I \setminus  I',$
there exists $g \in G$ and $t \in T$ such that $LM(tg)=LM(f)$
and $tS(g)=S(f).$ \label{def:S-GB}
\end{deftn}

\begin{rmk} \label{rmk:div_SGB_irred}
Unlike in Arri and Perry's paper \cite{Arri-Perry}, we ask for $tS(g)=S(f)$ 
instead of the weaker condition $S(tg)=S(f).$
The main reason is to avoid misshapes like that of Example \ref{exp:prob_sign}.
We can nevertheless remark that thanks to
Lemma \ref{lem:sign_comp_mult}, then $tS(g)=S(f)$ implies that
$S(tg)=tS(g)=S(f).$
\end{rmk}

We prove in this Section that tropical $\mathfrak{S}$-Gröbner bases are 
tropical Gröbner bases, allowing one of the main ideas of the F5 algorithm:
compute tropical $\mathfrak{S}$\textbf{-Gröbner basis} instead of
tropical Gröbner basis. 

To that intent, we use the following two propositions.

\begin{prop} \label{prop:S-irreduc}
A polynomial $f$ is $\mathfrak{S}$-irreducible iff $S(f)=\Phi (LM(f)).$
\end{prop}
\begin{proof}
By definition, if $S(f)=\Phi (LM(f)),$ then clearly $f$ is $\mathfrak{S}$-irreducible.
Regarding to the converse, if $S(f)>_{sign}\Phi (LM(f)),$ then there exists $g \in I$ such that
$LM(g)=LM(f)$ and $S(g)=\Phi (LM(f)),$ and then $g$ $\mathfrak{S}$-reduces $f$.
\end{proof}

\begin{prop} \label{prop:reduction_by_S_basis}
If $G$ is a tropical $\mathfrak{S}$-Gröbner basis, then for any nonzero $f \in I \setminus I',$ there exists $g \in G$ and $t \in T$ such that:
\begin{itemize}
\item $LM(tg)=LM(f)$
\item $S(tg)=tS(g)=S(f)$ if $f$ is irreducible, and $S(tg)=tS(g)<_{sign}S(f)$ otherwise.
\end{itemize}
Hence, there is an $\mathfrak{S}$-reductor for $f$ in $G$ if $f$ is not irreducible.
\end{prop}
\begin{proof}
If $f$ is irreducible, this is a result of Definition \ref{def:S-GB}.

Let us assume that $f$ is not $\mathfrak{S}$-irreducible. We take $h$ $\mathfrak{S}$-irreducible such that
$LM(h)=LM(f).$ Thanks to Proposition \ref{prop:S-irreduc}, it exists, and $S(h)<_{sign}S(f).$ So $h$ is an $\mathfrak{S}$-reductor of 
$f.$ There are then $t \in T$ and $g \in G$ such that $t LM(g)=LM(h)=LM(f)$ and $tS(g)=S(tg)=S(h).$
With Proposition \ref{prop:S-irreduc}, $tg$ is irreducible. The result is proved.
\end{proof}

We can now prove the desired connection between tropical 
$\mathfrak{S}$-Gröbner bases
and tropical Gröbner bases.

\begin{prop}
If $G$ is a tropical $\mathfrak{S}$-Gröbner basis, 
then $G$ is a tropical Gröbner basis of $I,$ for $<.$
\end{prop}
\begin{proof}
Let $t \in LM(I) \setminus LM(I').$ With Proposition \ref{prop:phi_bijection}, there exist 
$\sigma \in NS(Syz_{f_1})$ such that $\Phi^{-1}(\sigma)=t$
and $f \in I \setminus I'$ such that $LM(f)=t$ and $S(f)=\sigma.$
By Proposition \ref{prop:reduction_by_S_basis}, 
there exists $g\in G,$ $u \in T$ such that $LM(ug)=LM(f)=t.$
Hence, the span of $\{ LM(g), \: g \in G \}$ contains 
$LM(I) \setminus LM(I'),$
and $G \supset G'.$ Therefore $G$ is a tropical Gröbner basis of $I.$
\end{proof}

And we can also prove some finiteness result on tropical $\mathfrak{S}$-Gröbner bases,
which can be usefully applied to the problem of the
termination of the F5 algorithm.

\begin{prop} \label{prop:finiteness_S_basis}
Every tropical $\mathfrak{S}$-Gröbner basis contains a finite 
tropical $\mathfrak{S}$-Gröbner basis.
\end{prop}
\begin{proof}
Let $G=\{g_i \}_{i \in L}$ be a tropical $\mathfrak{S}$-Gröbner basis.
Let 
\begin{equation*}
      V : \: \:      
     \begin{array}{ccl}
       G & \rightarrow & T \oplus T\\
      g_i & \mapsto & (LM(g_i),S(g_i))\\
     \end{array}
\end{equation*}
be a mapping. By Dickson's lemma, there exists some finite set $J \subset L$
such that the monoid generated by the image of $V$ is generated by
$H_0=\{g_i \}_{i \in J}.$
We claim that $H=H_0 \cup G'$ is a finite tropical
 $\mathfrak{S}$-Gröbner basis.
We have taken the union with $G'$ to avoid any issue with $G'$ being non-minimal.
Let $f \in I \setminus I'$ be an $\mathfrak{S}$-irreducible polynomial.
Since $G$ is a tropical $\mathfrak{S}$-Gröbner basis, there exists 
$g_i \in G$ and $t \in T$ such that $tS(g_i)=S(tg_i)=S(f)$ and 
$tLM(g_i)=LM(tg_i)=LM(f).$
If $i \in J,$ we are fine.
Otherwise, there exists some $j \in J$ and $u,v \in T$ such that
$LM(g_i)=uLM(g_j)$ and $S(g_j)=vS(g_j).$
Three cases are possible.
If $u=v,$ then $t'=ut \in T$ satisfies $t' LM(g_j)=LM(f)$ and $t'S(g_j)=S(f)$ and we are fine.
If $u<_{sign}v$ then $t'=ut \in T$ satisfies $t' LM(g_j)=LM(f)$ and $t'S(g_j)<S(f),$ contradicting
the hypothesis that $f$ is $\mathfrak{S}$-irreducible.
If $u>_{sign}v,$ then for $t'=vt \in T,$ we can take some $\alpha \in k$ such that
$p=f-\alpha t'g_j$ satisfies $LM(p)=LM(f)$ but $S(p)<S(f),$ contradicting the irreducibility
of $f.$
As a consequence, we have proved that $H$ is a tropical 
$\mathfrak{S}$-Gröbner basis.
\end{proof}

The elements of $H_0$ we have used are of special importance, hence we give them a special name.

\begin{deftn}
We say that a non-zero polynomial  $f \in I \setminus I',$ $\mathfrak{S}$-irreducible (with respect to $S(f)$), is
\textbf{primitive $\mathfrak{S}$-irreducible} if there are no polynomials $f' \in I \setminus I'$
and terms $t \in T \setminus \{1 \}$ such that $f'$ is $\mathfrak{S}$-irreducible,
$LM(tf')=LM(f)$ and $S(tf')=S(f).$
\end{deftn}

The proof of Proposition \ref{prop:finiteness_S_basis} implies that we can obtain a 
finite tropical $\mathfrak{S}$-Gröbner basis by keeping a subset of primitive $\mathfrak{S}$-irreducible
polynomials with different leading terms.
Also, it proves there exists a finite tropical $\mathfrak{S}$-Gröbner basis with only primitive
$\mathfrak{S}$-irreducible polynomials (in its $I \setminus I'$ part).

\section{Linear algebra and tropical $\mathfrak{S}$-Gröbner bases}

When the initial polynomials from which we would like to compute a Gröbner basis
are homogeneous, the connection between linear algebra and Gröbner bases is well
known.

\begin{deftn}
Let $c_{n,d}= \binom{n+d-1}{n-1},$ and $B_{n,d}=(x^{d_i})_{1 \leq i \leq c_{n,d}} $ be the monomials of $A_d.$
Then for $f_1, \dots, f_s \in A$ homogeneous polynomials, with $\vert f_i \vert = d_i$, and $d \in \mathbb{N}$, we define $Mac_d(f_1, \dots, f_s)$ to be the matrix whose rows are the polynomials $x^{ \alpha_{i,j}} f_i$ written in the basis $B_{n,d}$ of $A_d$.

We note that 
$Im(Mac_d(f_1, \dots, f_s))=\left\langle f_1, \dots , f_s \right\rangle \cap A_d.$
\end{deftn}

\begin{theo}[\cite{Lazard:1983}] \label{echelon}
For an homogeneous ideal $I=\left\langle f_1,\dots , f_s \right\rangle, $ $(f_1,\dots , f_s)$ is a Gröbner basis of $I$ for a monomial order $\leq_0$ if and only if: for all $d \in \mathbb{N}$,
written in a decreasingly ordered $B_{n,d}$ (according to $\leq_0$), $Mac_d(f_1, \dots, f_s)$ contains an \textit{echelon basis} of $Im(Mac_d(f_1, \dots, f_s))$.
\end{theo}

By \textit{echelon basis}, we mean the following
\begin{deftn}
Let $g_1,\dots, g_r$ be homogeneous polynomials of degree $d$. Let $M$ be the matrix whose $i$-th row is the row vector corresponding to $g_i$ written in $B_{n,d}$.
Then we say that $(g_1,\dots, g_r)$ is an \textit{echelon basis} of $Im(M)$ if there is a permutation matrix $P$ such that $PM$ is under row-echelon form.
\end{deftn}

In other words, $G=(g_1,\dots,g_s)$ is a Gröbner basis of $I$ if and only if for all $d,$ 
an echelon (linear) basis
of $I_d$ is contained in the set 
$\{ x^\alpha g_i, i \in \llbracket 1,s \rrbracket, \: x^\alpha \in T, \: \vert x^\alpha g_i\vert=d  \}.$  

We have an analogous property for tropical Gröbner bases and tropical
 $\mathfrak{S}$-GB. It follows from the study in \cite{Vaccon:2015}
 of a tropical Matrix-F5 algorithm.
 We first need to adapt to the tropical setting the definitions
 of row-echelon form and echelon basis.

\begin{deftn}[Tropical row-echelon form]
Let $M$ be a $l \times m$ matrix which is a Macaulay matrix, written in the basis $B_{n,d}$
of the monomials of $A$ of degree $d.$
We say that $(P,Q) \in GL_n(k) \times GL_m(k),$ 
$Q$ being a permutation matrix,
realize a tropical row-echelon form of $M$ if:
\begin{enumerate}
\item $PMQ$ is upper-triangular and under row-echelon form.
\item The first non-zero coefficient of a row corresponds to the
leading term of the polynomial corresponding to this row.
\end{enumerate} \label{def:trop_row_echelon_form}
\end{deftn}
We can then define a \textit{tropical echelon basis}:
\begin{deftn}[Tropical echelon basis]
Let $g_1,\dots, g_r$ be homogeneous polynomials of degree $d$. Let $M$ be the matrix whose $i$-th row is the row vector corresponding to $g_i$ written in $B_{n,d}$.
Then we say that $g_1,\dots, g_r$ is a \textit{tropical echelon basis}
 of a vector space $V \subset A_d$ 
 if there are two permutation matrices $P,Q$
  such that $PMQ$ realizes a tropical row-echelon form of $M$
  and $span(Rows(M))=V.$
\end{deftn}

This can be adapted to the natural filtration of the vector space
$I$ by $(I_{d,x^\alpha})_{d,x^\alpha}$ we have defined in \ref{def:filtration}.

\begin{theo} \label{thm:echelon}
Suppose that $G$ is a set of $\mathfrak{S}$-irreducible homogeneous
 polynomials of the homogeneous ideal $I$ such that 
 $\{g \in G, \: S(g)=0 \}=G'.$
Then $G$ is a tropical $\mathfrak{S}$-Gröbner basis of $I$ 
if and only if \textbf{
for all} $x^\alpha \in T,$ 
taking $d=\vert x^\alpha \vert + \vert f_1 \vert,$ the set
 \[\{ x^\beta g, \text{ irreducible s.t. } g \in G, \: x^\beta \in T, \vert x^\beta g \vert =d, \: x^\beta S(g)=S(x^\beta g) \leq x^\alpha  \}\]
contains a tropical echelon basis of $I_{d, \leq x^\alpha}.$   
\end{theo}
\begin{proof}
Using Proposition \ref{prop:S-irreduc}, it is clear that if $G$ satisfy the above-written condition, 
then it satisfies Definition \ref{def:S-GB} of tropical $\mathfrak{S}$-GB.
The converse is also easy using Proposition \ref{prop:reduction_by_S_basis} on an echelon basis
of $I_{d, \leq x^\alpha}$ and remarking that to get an echelon basis, it is enough
to reach all the leading monomials of $I_{d, \leq x^\alpha}.$
\end{proof}

An easy consequence of the previous theorem is the following result of existence:

\begin{prop} \label{prop:S-GB-exist}
Given $G'$ and $f_1,$ consisting of homogeneous polynomials, there exists
a tropical $\mathfrak{S}$-GB for $I=\left\langle G',f_1 \right\rangle.$
\end{prop}
\begin{proof}
It is enough to compute a tropical echelon basis 
for all the $I_{d, \leq x^\alpha},$
by tropical row-echelon form computation (see \cite{Vaccon:2015}),
 and take the set of all these polynomials.
\end{proof}

Even with Proposition \ref{prop:finiteness_S_basis}, the idea of the proof
of Proposition \ref{prop:S-GB-exist} is not enough to obtain an efficient
algorithm. This is why we introduce the F5 criterion and design an
F5 algorithm.

\section{F5 criterion}

In this section, we explain a criterion, the F5 criterion, which 
yields an efficient algorithm 
to compute tropical Gröbner bases.

We need a slightly different notion of $S$-pairs, called here normal pairs.

\begin{deftn}[Normal pair]
Given $g_1,g_2 \in I,$ not both in $I',$ let $Spol(g_1,g_2)=u_1 g_1-u_2 g_2$ be the $S$-polynomial of
$g_1$ and $g_2,$ where $u_i = \frac{lcm(LM(g_1),LM(g_2))}{LT(g_i)}.$ We say that $(g_1,g_2)$
is a \textbf{normal pair} if:
\begin{enumerate}
\item the $g_i$'s are primitive $\mathfrak{S}$-irreducible polynomials.
\item $S(u_i g_i)=LM(u_i)S(g_i)$ for $i=1,2.$
\item $S(u_1 g_1) \neq S(u_2 g_2).$
\end{enumerate} \label{def:normal_pair}
\end{deftn}

\begin{rmk}
With this definition, if $(g_1,g_2)$ is a normal pair,
using Lemma \ref{lem:addition_signature},
$S(Spol(g_1,g_2))=\max (S(u_1 g_1),S(u_2 g_2))$
holds. Moreover, if $S(u_1 g_1) > S(u_2 g_2)$ then $u_1 \neq 1$ as if otherwise, $g_2$ would be
an $\mathfrak{S}$-reductor of $g_1.$ Therefore $S(Spol(g_1,g_2))>max(S(g_1),S(g_2)).$
\end{rmk}

\begin{theo}[F5 criterion]
Suppose that $G$ is a set of $\mathfrak{S}$-irreducible homogeneous polynomials of $I$ such that:
\begin{enumerate}
\item $\{g \in G, \: S(g)=0 \}=G'.$
\item if $f_1 \notin I',$ there exists $g \in G$  such that $S(g)=1.$
\item for any $g_1,g_2 \in G$ such that $(g_1,g_2)$ is a normal pair, there exists
$g \in G$ and $t \in T$ such that $tg$ is $\mathfrak{S}$-irreducible and 
$tS(g)=S(tg)=S(Spol(g_1,g_2)).$
\end{enumerate}
Then $G$ is a $\mathfrak{S}$-Gröbner basis of $I.$ \label{thm:F5crit}
\end{theo}
\begin{rmk}
The converse result is clearly true.
\end{rmk}
\begin{rmk}
The $g$ given in the second condition is primitive $\mathfrak{S}$-irreducible, by definition and
using Lemma \ref{lem:deg_of_sign}.
\end{rmk}

Theorem \ref{thm:F5crit} is an analogue of the Buchberger criterion
for tropical $\mathfrak{S}$-Gröbner bases. To prove it,
we adapt the classical proof of the Buchberger criterion.
We need three lemmas, the first two being very classical.

\begin{lem} \label{lem:Spoly_writing}
Let $P_1, \dots, P_r \in A,$ $c_1, \dots, c_r \in k$ and
$\beta$ a term in $A$, $\sigma \in T$  be such that for all
$i$ $LC(P_i)=1,$ $LT(c_i P_i)=\beta,$ $P_i \in I$ and $S(P_i) \leq \sigma.$
Let $P=c_1 P_1 +\dots+c_r P_r.$
If $LT(P)< \beta,$ then there exist some $c_{i,j} \in k$ such that
$P=\sum_{i,j} c_{i,j} Spol(P_i,P_j)$ and $LT(c_{i,j} Spol(P_i,P_j))<\beta.$
\end{lem}

\begin{lem} \label{lem:factor_Spoly}
Let $x^\alpha, x^\beta, x^\gamma, x^\delta \in T$ and $P,Q \in A$
 be such that $LM(x^\alpha P)=LM(x^\beta Q)=x^\gamma$
 and $x^\delta=lcm(LM(P), LM(Q)).$
 Then \[Spol(x^\alpha P, x^\beta Q)=x^{\gamma - \delta} Spol(P,Q). \]
\end{lem}

\begin{lem} \label{lem:SGB_rewriting}
Let $G$ be an $\mathfrak{S}$-Gröbner basis of $I$ up to signature $<\sigma \in T.$ Let $f \in I,$ homogeneous of degree $d$, be such that $S(f) \leq \sigma.$
Then there exist $r \in \mathbb{N},$ $g_1,\dots, g_r \in G,$
$Q_1,\dots, Q_r \in A$ such that for all $i$ and $x^\alpha$ a monomial of $Q_i,$
$S(x^\alpha g_i)=x^\alpha S(g_i) \leq \sigma$ and 
$LT(Q_i g_i) \leq LT(f).$
The $x^\alpha S(g_i)$'s are all distinct, when non-zero.
\end{lem}
\begin{proof}
It is clear by linear algebra. One can form
a Macaulay matrix in degree $d$ whose rows corresponds to polynomials
$\tau g$ with $\tau \in T, g \in G$ such that $S(\tau g)=\tau S(g) \leq \sigma.$ Only one per non-zero signature, and each of them reaching an 
element of $LM(I_{d, \leq \sigma}).$
It is then enough to stack $f$ at the bottom of this matrix
and perform a tropical LUP form computation (see Algorithm \ref{algo:trop LUP})
to read the $Q_i$ on the reduction of $f.$
\end{proof}
We can now provide a proof of Theorem \ref{thm:F5crit}.
\begin{proof}
We prove this result by induction on $\sigma \in T$
such that $G$ is an $\mathfrak{S}$-GB up to $\sigma.$
It is clear for $\sigma =1.$

Let us assume that $G$ is an $\mathfrak{S}$-GB up to signature $<\sigma$
for some some $\sigma \in T.$
We can assume that all $g \in G$ satisfy $LC(g)=1.$
Let $P \in I$ be irreducible and such that $S(P)=\sigma.$
We prove that there is $\tau \in  T, g\in G$ such that
$LM(P)=LM(\tau g)$ and $\tau S(g)=\sigma.$

Our second assumption for $G$ implies that there exist at least one primitive
$\mathfrak{S}$-irreducible $g \in G$ and some $\tau \in T$ such that 
$\tau S(g)=S(f)=\sigma.$
If $LM(\tau g)=LM(f)$ we are done.
Otherwise, by Lemma \ref{lem:baisse_signature}, 
there exist some $a, b \in k^*$ such that
$S(a f+b \tau g)=\sigma'$ for some $\sigma'<_{sign}\sigma.$

We can apply Lemma \ref{lem:SGB_rewriting} 
to $af+ b \tau g$ and obtain that 
there exist $r\in \mathbb{N},$ $Q_i \in A,$ $g_i \in G$ such that
$P=\sum_{i=1}^r Q_i g_i,$ $LT(Q_i g_i) \leq P$ and for all
$i,$ and $x^\gamma$ monomial of $Q_i,$
the $x^\gamma S(g_i)=S(x^\gamma g_i) \leq_{sign} \sigma$ are all
distinct. 
We remark that $LT(P) \leq \max_i (LT(g_i Q_i)).$
We denote by $m_i:=LT(g_i Q_i).$

Moreover, we can assume that all the $g_i$'s are primitive $\mathfrak{S}$-irreducible.
Indeed, if among them some $g'$ is not primitive $\mathfrak{S}$-irreducible,
then there exists $h_0,t_0$ in $I \times T\setminus \left\lbrace 1 \right\rbrace$ such that
$h_0$ is $\mathfrak{S}$-irreducible and $LM(t_0 h_0)=LM(g')$ and $t_0S( h_0)=S(g')=S(t_0 h_0).$
We have $S(h_0) \leq_{sign} S(g')<_{sign}\sigma.$ Hence, we can apply the $\mathfrak{S}$-GB property for
$h_0$ and we obtain $g_0',\tau_0$ in $G \times T$ such that $LM(h_0)=LM(\tau_0 g_0')$ and $S(h_0)=S(\tau_0 g_0')=\tau_0 S(g_0').$
We then have $LM(g')=LM(t_0 \tau_0 g_0')$ and $S(g')=t_0 \tau_0 S( g_0')=S(t_0 \tau_0 g_0'),$ with $\deg(LM(g'))>\deg (LM(g_0')).$
As a consequence, this process can only be applied a finite number of times before we obtain a $g_k' \in G$ which is 
primitive $\mathfrak{S}$-irreducible and some $b \in T$ such that
$LM(b g_k')=LM(\tau g)$
and 
$bS( g_k')=S(b g_k')=\sigma' <_{sign} \sigma=S(\tau g).$
Thus, we can assume that all the $g_i$'s are primitive $\mathfrak{S}$-irreducible. 

Among all such possible way of writing $P$ as $\sum_{i=1}^r Q_i g_i,$
we define $\beta$ as the \textbf{minimum} of the $\max_i (LT(g_i Q_i))$'s.
$\beta$ exists thanks to Lemma 2.10 of \cite{CM:2013}.

If $LT(P)=\beta,$ then we are done. Indeed, there is then some $i$ and $\tau$
in the terms of $Q_i$ such that $LT(\tau g_i)=\beta$
and $S(\tau g_i) \leq \sigma.$

We now show that $LT(P)<\beta$ leads to a contradiction.

In that case, we can write that:
\begin{align*}
P&=\sum_{m_i = \beta} Q_i g_i +\sum_{m_i < \beta} Q_i g_i, \\
&=\sum_{m_i = \beta} LT(Q_i) g_i +\sum_{m_i = \beta} (Q_i-LT(Q_i)) g_i+\sum_{m_i < \beta} Q_i g_i.\\
\end{align*}
As $LT(P) <\beta$ and this is also the case for the two last summands in
the second part of the previous equation, 
$LT(\sum_{m_i = \beta} LT(Q_i) g_i)<\beta.$
We write $LT(Q_i)=c_i x^{\alpha_i},$ with $c_i \in k$
and $\beta=c_0 x^{\beta_0}$ for some $c_0 \in k.$
Thanks to Lemma \ref{lem:Spoly_writing} and \ref{lem:baisse_signature},
there are some $c_{j,k} \in k$
and $x^\beta_{j,k}=lcm(LM(g_j),LM(g_k))$ such that
\[\sum_{m_i = \beta} LT(Q_i) g_i=\sum_{m_i, m_j=\beta} c_{j,k} x^{\beta_0-\beta_{j,k}} Spol(g_j,g_k). \]
Moreover, we have for all $j,k$ involved,
$S(x^{\beta_0-\beta_{j,k}} Spol(g_j,g_k))\leq \sigma$
and $LT(c_{j,k} x^{\beta_0-\beta_{j,k}} Spol(g_j,g_k))<\sigma.$

If there is $j,k$ such that $S(Spol(g_j,g_k))= \sigma,$
then, by the way the $Q_i$ were chosen (distinct signatures, multiplicativity
of the signatures), the pair $(g_j,g_k)$ is normal
and the third assumption is enough to conclude.

Otherwise, we have for all $j,k$ involved, $S(Spol(g_j,g_k))< \sigma.$
We can apply Lemma \ref{lem:SGB_rewriting}
to obtain 
$c_{j,k} x^{\beta_0-\beta_{j,k}} Spol(g_j,g_k)=\sum_i Q_i^{j,k} g_i$
such that for all $i$ and $x^\gamma$ monomial of $Q_i^{j,k}$
$LT(Q_i^{j,k} g_i)<\beta$ and $x^\gamma S(g_i)=S(x^\gamma g_i)\leq \sigma.$

All in all, we obtain some $\tilde{Q}_i$ such that 
$P=\sum_i \tilde{Q}_i g_i$ and for all $i$ $LT(\tilde{Q}_i g_i)<\beta.$
This contradicts with the definition of $\beta$ as a minimum.
So $LT(P)=\beta,$ which concludes the proof.
\end{proof}

\begin{rmk}
This theorem holds for $\mathfrak{S}$-GB up to a given signature or, as we work with homogeneous entry polynomials, for
$\mathfrak{S}$-GB up to a given degree (\textit{i.e.} $d-\mathfrak{S}$-GB).
\end{rmk}

\section{A tropical F5 algorithm}


Theorem \ref{thm:F5crit} gives a first idea on how to do a Buchberger-style algorithm for $\mathfrak{S}$-GB.
Yet, deciding in advance whether a pair is a normal pair does not seem to be easy.
Indeed, the second condition require some knowledge on $LM(Syz_{f_1})$ which we usually do not have.
There are two natural ways to handle this question:
Firstly, we could keep track during the algorithm of the syzygies encountered, and use a variable $L$ as a place holder for 
their leading monomials. The second condition can then be replaced by 
$LM(u_i) S(g_i) \notin \left\langle L \right\rangle. $
This is what is used in \cite{Arri-Perry}.
Another way is to only consider the trivial syzygies. This amounts to take 
$\left\langle L \right\rangle = LM(I') $
and use the same replacement for the second condition. This is what is used in
\cite{F5} and \cite{JCFnoJNCF}.

We opt for the \textbf{second choice} (only handling trivial syzygies). This give rise to
the notion of admissible pair.

\begin{deftn}[Admissible pair] \label{def:adm_pairs}
Given $g_1,g_2 \in I,$ not both in $I',$ let $Spol(g_1,g_2)=u_1 g_1-u_2 g_2$ be the $S$-polynomial of
$g_1$ and $g_2.$ We have $u_i = \frac{lcm(LM(g_1),LM(g_2))}{LT(g_i)}.$ We say that $(g_1,g_2)$
is an \textbf{admissible pair} if:
\begin{enumerate}
\item the $g_i$'s are primitive $\mathfrak{S}$-irreducible polynomials.
\item if $S(g_i)\neq 0,$ then $LM(u_i)S(g_i) \notin LM(I').$
\item $S(u_1 g_1) \neq S(u_2 g_2).$
\end{enumerate}
\end{deftn}

We can then remark that handling admissible pairs instead of normal pairs is harmless,
as the latter is a subset of the former.

\begin{lem} \label{lem:crit_adm_pairs}
If a set $G$ satisfies the conditions of Theorem \ref{thm:F5crit} for all its admissible pairs then it is an $\mathfrak{S}$-GB.
\end{lem}

In the general case, $LM(Syz_{f_1})$ is not known in advance. However,
it can be determined \textbf{inductively} on signatures.
This is how the following algorithm will proceed.
From a polynomial $g,$ the signature of $x^\alpha g$ will
be \textbf{guessed} as $x^\alpha S(g),$ and after an $\mathfrak{S}$-GB
up to signature $<x^\alpha S(g)$ is computed, we can decide whether
$S(x^\alpha g)=x^\alpha S(g),$ or else $x^\alpha g$ happens
to be reduced to zero.
In the following, we certify inductively whether for a processed
$x^\alpha g,$ the \textbf{guessed} signature $x^\alpha S(g)$
equals the \textbf{true} signature $S(x^\alpha g).$
Similarly, \textbf{guessed} admissible pairs are inductively certified to
be \textbf{true} admissibles pairs or not
once condition 3 of \ref{def:adm_pairs} is certified.
Using this idea, we provide a first version of an 
F5 algorithm in Algorithm
\ref{F5_algo}, using  
Algorithm \ref{algo:SymbPreproc} for Symbolic Preprocessing.

\IncMargin{1em}
\begin{algorithm} 
\DontPrintSemicolon

 \SetKwInOut{Input}{input}\SetKwInOut{Output}{output}

 \Input{$G'$ a tropical GB of $I'$ consisting of homogeneous polynomials, $f_1$ an homogeneous polynomial, not in $I'$}
 \Output{A tropical $\mathfrak{S}$-GB $G$ of $I'+\left\langle f_1 \right\rangle$}

$G \leftarrow \{ (0,g) \text{ for g in } G' \}$ ;  \;
$f \leftarrow f \mod G'$ (classical reduction) ; \;
$G \leftarrow G \cup \{ (1,f) \}$ ;  \;
$B \leftarrow \{ \text{guessed admissible pairs of } G \}$ ; \;
$d \leftarrow 1$ ; \;
\While{$B \neq \emptyset$}{
$P$ \textbf{receives} the pop of the guessed admissible pairs in $B$ of degree $d$ ; \;
\textbf{Write} them in a Macaulay matrix $M_d$, along with their $\mathfrak{S}$-reductors obtained from $G$ (one per non-zero signature) by \textbf{Symbolic-Preprocessing}$(P,G)$ (Algorithm \ref{algo:SymbPreproc}); \;
\textbf{Apply} Algorithm \ref{algo:trop LUP} to compute the $U$ in the tropical LUP form of $M$ (no choice of pivot) ; \;
\textbf{Add} to $G$ all the polynomials obtained from $\widetilde{M}$ that provide new leading monomial up to their signature ; \;
\textbf{Add} to $B$ the corresponding new admissible pairs ; \;
$d \leftarrow d+1$ ; \;
}			
Return $G$ ; \; 

 \caption{A first F5 algorithm} \label{F5_algo}
\end{algorithm}
\DecMargin{1em}

\begin{rmk}
\textbf{Only signature zero} is allowed to appear multiple times in the matrix in construction.
\end{rmk}
The reason is the following: because of Proposition
\ref{prop:S-irreduc}, for an irreducible polynomial
with a given signature, its leading monomial
is determined by its signature.
After performing the tropical row-echelon form computation,
all rows corresponds to irreducible polynomials, hence
two rows produced with the same signature are redundant:
either they will produce the same leading monomial or 
they would reduce to zero.

\IncMargin{1em}
\begin{algorithm}

\SetKwInOut{Input}{input}\SetKwInOut{Output}{output}

\Input{$P$, a set of admissible pairs of degree d and $G$, a $\mathfrak{S}$-GB up to degree $d-1$}
\Output{A Macaulay matrix of degree d}
$D \leftarrow$ the set of the \textbf{leading monomials} of the polynomials in $P$ \;
$C \leftarrow$ the set of the \textbf{monomials} of the polynomials in $P$ \;

$U \leftarrow $ the polynomials of $P$ \;
\While{$C \neq D$}{
$m \leftarrow \max (C \setminus D)$ \;
$D \leftarrow D \cup \{ m \}$ \;
$V \leftarrow \emptyset$ \;
\For{$g \in G$}{
\If{$LM(g) \mid m$}{
$V \leftarrow V \cup \{(g, \frac{m}{LM(g)}) \}$ \;}
}
$(g, \delta) \leftarrow$ the element of $V$ with $\delta g$ of \textbf{smallest signature}, with tie-breaking by taking minimal $\delta$ (for degree then for $\leq_{sign}$) \;  
$U \leftarrow U \cup \{ \delta g \}$ \;
$D \leftarrow D \cup \{ \text{monomials of } \delta g \}$ \;
}
$M \leftarrow$ the polynomials of $U,$ written in Macaulay matrix of degree $d$
and ordered by increasing signature, with no repetition of signature outside of signature $0$ (choosing smallest leading monomial to break a tie of signature) \;			
Return $M$ \; 
 \caption{Symbolic-Preprocessing} \label{algo:SymbPreproc}
\end{algorithm}
\DecMargin{1em}

The tropical LUP form computation to obtain a row-echelon matrix,
with no choice of pivot, is described in Algorithm \ref{algo:trop LUP}.
See \cite{Vaccon:2015} for more details. The result we want to prove is then:

\IncMargin{1em}
\begin{algorithm} 
\DontPrintSemicolon

 \SetKwInOut{Input}{input}\SetKwInOut{Output}{output}

 \Input{$M$, a Macaulay matrix of degree $d$ in $A$, with $n_{row}$ rows and $n_{col}$ columns, and $mon$ a list of monomials indexing the columns of $M.$}
 \Output{$\widetilde{M}$, the $U$ of the tropical LUP-form of $M$}

$\widetilde{M} \leftarrow M$ ;  \;
\eIf{$n_{col}=1$ or $n_{row}=0$ or $M$ has no non-zero entry}{
				Return $\widetilde{M}$  ;\;
	}{			
\For{$i=1$ to $n_{row}$}{
\textbf{Find} $j$ such that $\widetilde{M}_{i,j}$ has the greatest term $\widetilde{M}_{i,j} x^{mon_j}$ for $\leq$ of the row $i$; \;
\textbf{Swap} the columns $1$ and $j$ of $\widetilde{M}$, and the $1$ and $j$ entries of $mon$; \;
By \textbf{pivoting} with the first row, eliminates the coefficients of the other rows on the first column; \;
\textbf{Proceed recursively} on the submatrix $\widetilde{M}_{i \geq 2, j \geq 2}$; \;}
Return $\widetilde{M}$; \;}

 \caption{The tropical LUP algorithm} \label{algo:trop LUP}
\end{algorithm}
\DecMargin{1em}

\begin{theo}
Algorithm \ref{F5_algo} computes an $\mathfrak{S}$-GB of $I.$
\end{theo}
\begin{proof}

\textbf{Termination:} Assuming correctness, after (theoretically) performing the algorithm for all degree $d$
 in $\mathbb{N}$, we obtain an $\mathfrak{S}$-GB. 
 Since by Proposition \ref{prop:finiteness_S_basis} all $\mathfrak{S}$-GB contain a finite $\mathfrak{S}$-GB
 then at some degree $d$ we have computed a finite $\mathfrak{S}$-GB.
 As a consequence, all $S$-pairs from degree $d+1$ to degree $2d$ (at most) will not yield any
 new polynomial in $G$ (no new leading monomial), and thus there will be no $S$-pair of 
 degree more than $2d,$ which proves the termination of the algorithm.

\textbf{Correctness:} 
We proceed by induction on the signature to prove that the result
of Algorithm \ref{F5_algo} is a tropical $\mathfrak{S}$-GB.
The result is clear for signature $\leq_{sign} 1.$

For the induction step, we assume that the result is proved up to signature $\leq_{sign} x^\alpha,$
with $\vert x^\alpha f_1 \vert =d.$
Let $x^\beta$ be the smallest guessed signature of $M_d$ of signature $>_{sign} x^\alpha.$

We first remark that if there are rows of guessed signature $>_{sign}x^\beta$ that
are of true signature $<_{sign}x^\beta$ then:
\textbf{1.} We can conclude that there is no normal pair popped from $B$ with second half of a pair with signature $x^\gamma$ 
such that $x^\alpha <_{sign} x^\gamma <_{sign}x^\beta$ because of condition 2 of Definition \ref{def:normal_pair} (which  prevents such signature to drop).
\textbf{2.} Using the F5 Criterion Theorem \ref{thm:F5crit}, it proves that we have in $G$ (and the rows of
$M_d$ up to signature $x^\alpha$ that are added to $G$) an $\mathfrak{S}$-GB up to signature $<_{sign}x^\beta.$
\textbf{3.} As a consequence, using Theorem \ref{thm:echelon} the Symbolic Preprocessing has produced exactly enough rows of guessed (and true) signature $<_{sign}x^\beta$ from $G$ to 
$\mathfrak{S}$-reduce the row of guessed signature $x^\beta.$
Indeed, since we have an $\mathfrak{S}$-GB up to $<_{sign} x^\beta,$ all necessary leading monomials could be
attained by product monomial-polynomial of $G$ with guessed signature $<_{sign} x^\beta$ or
through the echelon form up to $<_{sign}x^\beta$ of $M_d.$  
The last consequence is of course also true if there is no such row with a gap between the
guessed and the true signature. 

Two possibilities can occur for the result of the reduction of the row of guessed signature $x^\beta$:
\textbf{1.} The row reduces to zero. Then the signature $x^\beta$ is not possible. We then have in $G$ an
$\mathfrak{S}$-GB up to signature $\leq_{sign} x^\beta.$
\textbf{2.} The row does not reduce to zero. Then, depending on whether the reduced row
provide a new leading monomial for $I_{\leq_{sign} x^\beta},$ we add it to $G.$ We then have in $G$ an
$\mathfrak{S}$-GB up to signature $\leq_{sign} x^\beta.$
This concludes the proof by induction.
We then can apply the modified F5 Criterion, Lemma \ref{lem:crit_adm_pairs} to conclude that the output of 
Algorithm \ref{F5_algo} is indeed an $\mathfrak{S}$-GB.
\end{proof}

To conclude the proof of Theorem \ref{thm_intro}, 
the main result on the efficiency of the F5 algorithm is 
still valid for its tropical version:

\begin{prop}
If $f_1$ is not a zero-divisor in $A/I',$ then all the processed matrices
$M_d$ are (left-) injective. In other words, no row reduces to zero.
\end{prop}
\begin{proof}
In this case, $LM(Syz_{f_1})=LM(I').$
Hence, with the choice of rows of $M_d$
avoiding guessed signature in $LM(Syz_{f_1})$
no syzygy can be produced.
\end{proof}


\begin{rmk}[Rewritability]
Thanks to Theorem \ref{thm:F5crit}, it is possible to replace the polynomials
in $P$ in the call to \textbf{Symbolic-Preprocessing} on Line 8
of Algorithm \ref{F5_algo}. They can be replaced by any
other multiple of element of $G$ of the same signature.
Indeed, if one of them, $h,$ is of signature $x^\alpha,$ the algorithm
computes a tropical $\mathfrak{S}$-Gröbner basis up signature
$<x^\alpha.$ Hence, $h$ can be replaced by any other polynomial of same signature,
it will be reduced to the same polynomial.
By induction, it proves all of them can be replaced at the same time.
This paves the way for the Rewritten techniques of \cite{F5}.
The idea, as far as we understand it, is then to use 
the polynomial that has been the most reduced to produce a polynomial of 
signature $S(tg)$ for the upcoming reduction.
Taking the $x^\beta g$ ($g \in G)$ of signature $x^\alpha$ such that $g$
has the biggest signature possible is a first reasonnable idea.\footnote{Indeed,
such a $g$ is at first glance the most reduced possible.}
It actually can lead to a substantial reduction of the running time of the F5 algorithm.
\end{rmk}

\section{Implementation and numerical results}

A toy implementation of our algorithms 
in Sagemath \cite{Sage} 
is available
on \url{https://gist.github.com/TristanVaccon}.

\begin{rmk}
It is possible to apply Algorithm \ref{F5_algo} to compute a tropical
Gröbner basis
of $I$ given by $F=(P_1,\dots,P_s)$ by performing complete computation 
succesively for $(P_1),$ $(P_1,P_2),$ $(P_1,P_2,P_3), \dots$
adding a polynomial at a time playing the part $f_1$ played in the rest of the article.
As we only deal with homogeneous polynomials, it is also possible
to do the global computation degree by degree, and at a given degree
iteratively on the initial polynomials.
By indexing accordingly the signatures, as in \cite{F5},
the algorithm can be adapted straightforwardly. This is what has been chosen
in the implementation we have achieved.
\end{rmk}

We have gathered some numerical results in the following array. Timings are in seconds of CPU time.\footnote{Everything was performed in a guest Ubuntu 14.04
inside a Virtual Machine, with 4 processors and 29 GB of RAM.}
We have compared ours with that of the algorithms of
Chan and Maclagan in \cite{CM:2013} (in Macaulay 2)
and Markwig and Ren in \cite{MY:2015} (in Singular),
provided in \cite{MY:2015}. A dot means that
the computation could not complete.
Entry systems are homogenized. Base field is $\mathbb{Q}.$

\hspace{-.5cm}
\begin{tabular}{|c|c|c|c|c|c|c|c|c|}
\hline 
 & Katsura 3 & 4 & 5 & 6 & 7 & Cyclic 4 & 5 & 6 \\ 
\hline 
\cite{CM:2013} & $\leq 1$ &• & • & • & • & • & • & • \\
\hline 
\cite{MY:2015} & $\leq 1$ &$\leq 1$ & $\leq 1$ &•  & • & $\leq 1$ & $\leq 1$ & • \\ 
\hline 
Trop. F5 & $\leq 1$ & 5 & 74 & 513 & • & 4 & 353 & • \\ 
\hline 
\end{tabular} 

Loss in precision has also been estimated in the following setting.
For a given $p,$ we take three polynomials with random coefficients in $\mathbb{Z}_p$
(using the Haar measure)
in $\mathbb{Q}_p[x,y,z]$ of degree $2 \leq d_1 \leq d_2 \leq d_3 \leq 4.$
For any given choice of $d_i$'s, we repeat the experiment 50 times.
Coefficients of the initial polynomials are all given at some high enough
precision $O(p^N).$ Coefficients of the output tropical GB
are known at individual precision $O(p^{N-m}).$
We compute the total mean and max on those $m$'s on the obtained tropical GB.
Results are compiled in the following array as couples of mean and max, with $D=d_1+d_2+d_3-2$ the Macaulay bound.
\hspace{-.5cm}
\begin{tabular}{|c|c|c|c|c|c|c|c|}
\hline 
\tiny{$w=[0,0,0]$} & $D=4$ & 5 & 6 & 7 & 8 & 9 & 10  \\ 
\hline 
$p=2$ &  (.3,8) & (.4,11) & (.1,10) & (.1,11) & (.1,13) & (.1,9) & (.2,12)\\ 
\hline 
3 & (.1,4) & (.1,5) & (.2,7) & (.1,13) & (.1,16) & (.1,5)&(0,6) \\ 
\hline 
101 & (0,1) & (0,0) & (0,0) & (0,1) & (0,1) & (0,0)&(0,1) \\ 
\hline
65519 & (0,0) &(0,0)&(0,0) &(0,0)&(0,0) & (0,0)&(0,0) \\ 
\hline  
\end{tabular} \\
\hspace{-.5cm}
\begin{tabular}{|c|c|c|c|c|c|c|c|}
\hline 
\tiny{$w=[1,-3,2]$} & $D=4$ & 5 & 6 & 7 & 8 & 9 & 10  \\ 
\hline 
$p=2$ &  (.2,7) & (.5,12) & (3.3,45) & (3.5,29) & (3.7,24) & (4.8,85) & (4.8,86)\\ 
\hline 
3 & (1.5,13) & (1.2,9) & (4.2,20) & (3.6,19) & (4.3,22) & (6,33)&(5.8,43) \\ 
\hline 
101 & (.1,2) & (.1,3) & (.1,4) & (.1,4) & (0,3) & (.2,5)&(.3,6) \\ 
\hline
65519 & (0,0) &(0,0)&(0,0) &(0,0)&(0,0) & (0,0)&(0,0) \\ 
\hline  
\end{tabular} 

As for Tropical Matrix-F5, a weight differing from $w=[0,0,0]$
yields bigger loss in precision.
Regarding to precision in row-reduction, in F5, 
this weight always use the best pivot on each row.
For Matrix-F5, it is always the best pivot
available in the matrix.
In view of our data, we can observe that the loss in precision
for Tropical F5 on these examples,
even though it is, as expected, bigger, has remained
reasonnable compared to the one of \cite{Vaccon:2015} 
that allowed full choice of pivot.

\vspace{-.2cm}
\section{Future works}

In this article, we have investigated the main step
for a complete F4-style tropical F5 algorithm.
We would like to understand more deeply
the Rewritten critertion of \cite{F5}.
We would also like to understand the natural extension of our work
to a Tropical F4 and to a Tropical F5 for non-homogeneous entry polynomials.

\begin{small}

\bibliographystyle{plain}

\end{small}

\end{document}